\documentclass[10pt, conference]{IEEEtran}

\IEEEoverridecommandlockouts

\usepackage{cite}
\usepackage{amsmath,amssymb,amsfonts}
\usepackage{graphicx}
\usepackage{textcomp}
\usepackage{xcolor}
\usepackage{amsmath}
\usepackage{verbatim}
\usepackage{overpic}
\usepackage{algorithm}
\usepackage{algpseudocode}
\usepackage{graphics}
\usepackage{epsfig}
\usepackage{algorithm,algpseudocode}
\def\BibTeX{{\rm B\kern-.05em{\sc i\kern-.025em b}\kern-.08em
    T\kern-.1667em\lower.7ex\hbox{E}\kern-.125emX}}

\usepackage{amsthm}
\newtheorem{theorem}{Theorem}
\newtheorem{assumption}{Assumption}
\usepackage{nidanfloat}
\usepackage{changepage}
\usepackage{bbm}

\IEEEoverridecommandlockouts\IEEEpubid{\makebox[\columnwidth]{ 978-1-6654-3540-6/22~\copyright~2022 IEEE \hfill} \hspace{\columnsep}\makebox[\columnwidth]{ }}

\begin{document}

\title{UAV-Assisted Hierarchical Aggregation for Over-the-Air Federated Learning}

\author{\IEEEauthorblockN{Xiangyu Zhong, Xiaojun Yuan, \IEEEmembership{Senior Member, IEEE}, Huiyuan Yang, Chenxi Zhong}
		\IEEEauthorblockA{\text{Yangtze Delta Region Institute, University of Electronic Science and Technology of China, Huzhou, China} \\
		National Key Lab of Sci. and Tech. on Communications, University of Electronic Sci. and Tech. of China, Chengdu, China \\ 
		Emails: xyzh@std.uestc.edu.cn, xjyuan@uestc.edu.cn, \{hyyang, cxzhong\}@std.uestc.edu.cn}
}

\maketitle

\begin{abstract}
	With huge amounts of data explosively increasing on the mobile edge, over-the-air federated learning (OA-FL) emerges as a promising technique to reduce communication costs and privacy leak risks. However, when devices in a relatively large area cooperatively train a machine learning model, the attendant straggler issue will significantly reduce the learning performance. In this paper, we propose an unmanned aerial vehicle (UAV) assisted OA-FL system, where the UAV acts as a parameter server (PS) to aggregate the local gradients hierarchically for global model updating. Under this UAV-assisted hierarchical aggregation scheme, we carry out a gradient-correlation-aware FL performance analysis. We then formulate a mean squared error (MSE) minimization problem to tune the UAV trajectory and the global aggregation coefficients based on the analysis results. An algorithm based on alternating optimization (AO) and successive convex approximation (SCA) is developed to solve the formulated problem. Simulation results demonstrate the great potential of our UAV-assisted hierarchical aggregation scheme.
\end{abstract} 


\section{Introduction}
In the era of big data, the explosion of distributed edge data exposes the weaknesses of centralized machine learning (ML), which requests edge devices to upload local data to a central server for training, leading to huge communication costs and high data leakage risks. Federated learning (FL) has recently arisen as a promising distributed learning paradigm to tackle these challenges by replacing data transmission with model/gradient exchange \cite{konevcny2016federated}. Specifically, in each FL training iteration, the PS first broadcasts the global model parameters to the devices, and the latter then calculates their local gradients separately based on the local datasets and uploads the local gradients to the PS. Subsequently, the PS aggregates the received local gradients and updates the global model. 

However, due to the iterative nature, the huge communication cost is still a bottleneck for FL. A promising approach to reduce the communication cost of FL in the wireless edge is to  introduce the over-the-air computation (AirComp) technique \cite{nazer2007computation} into FL uplink \cite{yang2020federated}. By utilizing the superposition property of wireless signals, AirComp allows the uplink cost to be unchanged with the growth of the device number  (with fixed aggregation accuracy), thereby significantly reducing the communication cost of large FL systems. The corresponding FL paradigm is referred to as over-the-air FL (OA-FL).

Despite these advantages, the introduction of AirComp also brings some unique problems, such as the straggler issue \cite{zhu2019broadband, liu2021reconfigurable, zhong2021over}. Since AirComp requires local gradients to be aligned at the PS (i.e., the PS receives a linear combination of local gradients with desired coefficients), the devices with good channel conditions have to lower their transmitting powers to match the devices with relatively poor channel conditions (i.e., the stragglers) to ensure correct aggregation. Existing works like \cite{zhu2019broadband} and \cite{liu2021reconfigurable} discard the stragglers to relieve the straggler issue. However, such a coarse-grained discarding strategy will inevitably cast away some exploitable information, resulting in a learning performance loss. Ref. \cite{zhong2021over} suggests a soft aggregation approach to tackle the straggler issue without discarding devices, which, however, works well only in small service area cases (e.g., $100\times 100$ m$^2$). When it comes to a relatively large service area (e.g., $2\times 2$ km$^2$), it fails due to the large difference in path loss between nearby devices and those far from the PS (which results in the de facto discarding of remote devices).

Allowing PS mobility is a natural idea to perform OA-FL in a large service area since the PS can thus move close to the stragglers for better services. Therefore, in this paper, an unmanned aerial vehicle (UAV) is introduced as the PS. Due to the ground-to-air nature of UAV uplink communication, the device-UAV links are generally not blocked, potentially leading to good channel conditions \cite{lin2018sky,wu2018joint}. Besides, UAVs have extremely high mobility and can thus cover a large area quickly. There are already some works on UAV-assisted FL. For example, Ref. \cite{lim2021uav} employs a UAV as a mobile relay in the sky to assist FL; in \cite{donevski2021federated}, a UAV acts as an orchestrator, coordinating the transmissions and the learning schedule within a preset deadline. 
However, these works consider orthogonal FL uplink and thus have various limitations in device number, service area, and communication condition, which means the state-of-art design is far from fully unleashing the potential of UAV-assisted FL.

\begin{figure}
	[t]
	\centering
	\includegraphics[height=0.205\textheight]{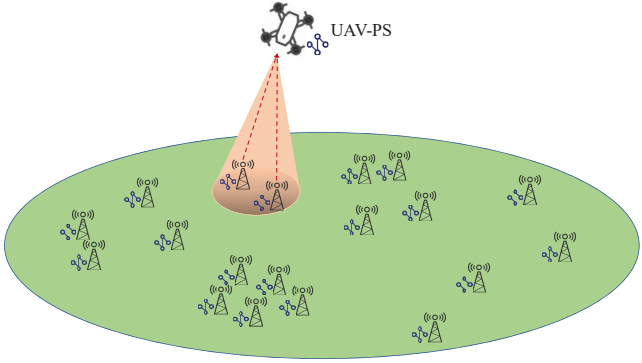}
	\caption{UAV-assisted OA-FL in a large service area.}
	\label{modelscheme}
	\vspace{-1.0em}
\end{figure}

In this paper, we make the first attempt to employ the UAV to assist the OA-FL system by designing a hierarchical over-the-air aggregation scheme.
As shown in Fig. \ref{modelscheme}, we consider a UAV-assisted OA-FL system, where the UAV acts as the PS, denoted by UAV-PS, to serve devices distributed in a relatively large area. To avoid serious straggler issues, we restrict the UAV-PS to serve nearby devices, i.e., only aggregate the local gradients from nearby devices over the air. Thus the UAV-PS must fly across its large service area to serve more devices. After the UAV-PS completes one round, it further aggregates the received partially aggregated local gradients to obtain a noisy version of the desired global gradient for global model updating. We call this two-step aggregation hierarchical aggregation. Under this hierarchical gradient aggregation scheme, we carry out a gradient-correlation-aware FL convergence analysis, which leads us to formulate a mean squared error (MSE) minimization problem to tune the UAV trajectory and the global aggregation coefficients. An algorithm based on alternating optimization (AO) and successive convex approximation (SCA) is proposed to solve the formulated problem. 
Simulation results demonstrate that our proposed scheme achieves robust performance improvement in complicated FL scenarios compared with existing solutions.


\section{System Model}
In this section, we detail the underlying models and the hierarchical aggregation scheme.

\subsection{Federated Learning Model}

We consider an FL system where $M$ devices collaboratively train a machine learning model assisted by a UAV-PS. Let $\mathbf{w} \in \mathbb{R}^D$ be the parameter vector of the machine learning model and $\mathcal{D}_m$ be the dataset on device $m$. Also let $Q_m \triangleq |\mathcal{D}_m|$ and $Q \triangleq \sum_{m=1}^M Q_m$. Then the global loss function can be written as
\begin{equation}
	\label{LossFunc}
	F(\mathbf{w}) = \sum_{m=1}^{M} b_m F_m \left( \mathbf{w}\right)
\end{equation}
with the local loss function
\begin{equation}
	\label{localLF}
	F_m(\mathbf{w}) = \frac{1}{Q_m} \sum_{q=1}^{Q_m} f ( \mathbf{w}; \boldsymbol{\xi}_{m,q} ),
\end{equation}
where $f ( \mathbf{w}; \boldsymbol{\xi}_{m,q} )$ denotes the sample-wise loss function based on the $q$-th data sample $\boldsymbol{\xi}_{m,q} \in \mathcal{D}_m$ and $b_m \triangleq Q_m/Q$.

FL requires, say, $T$ training rounds for convergence. The $t$-th training round consists of four steps:

\begin{itemize}
	\item \textit{Global model broadcasting}: The UAV-PS broadcasts the current global model parameters, $\mathbf{w}^{(t)}$, to $M$ devices in the service area. 
	\item \textit{Local gradient computation}: Each device computes its local gradient by
	\begin{gather}
		\mathbf{g}_{m}^{(t)}=\nabla F_m(\mathbf{w}^{(t)}),
	\end{gather}
	where $\nabla F_m(\mathbf{w}^{(t)})$ denotes the gradient of the local loss function.
	\item \textit{Hierarchical gradient aggregation}: The UAV-PS flies across the service area to hierarchically aggregate the local gradients to obtain $\mathbf{\hat{g}}^{(t)}$, which is a noisy version of the desired global gradient
	\begin{gather}
			\mathbf{g}^{(t)}\triangleq\sum_{m=1}^{M}b_m\mathbf{g}_{m}^{(t)}\label{idealtri}.
	\end{gather}
	\item \textit{Global model updating}: The UAV-PS updates the global model $\mathbf{w}^{(t)}$ by
	\begin{gather}
		\mathbf{w}^{(t+1)} = \mathbf{w}^{(t)} - \eta\mathbf{\hat{g}}^{(t)},
	\end{gather}
	where $\eta$ denotes the learning rate.
\end{itemize}

\subsection{UAV Channel Model}
\label{subsec_channel}
In this paper, we consider a three-dimensional Cartesian coordinate system. Assume that the $m$-th device is located on the ground with horizontal coordinate $\mathbf{v}_m=[x_m, y_m]^\top, m\in M$. Also assume that the UAV-PS flies at a fixed height $z$ above the ground with horizontal coordinate $\mathbf{u}(\tau)=[x(\tau),y(\tau)]^\top, 0\le\tau\le \Delta t$, at instant $\tau$, where  $\Delta t$ is the duration required by the UAV-PS for flying one round. In each training round, the UAV trajectory satisfies the following constraints:
\begin{align}
	\mathbf{u}(0)&=\mathbf{u}(\Delta t),\label{u0}\\
	\lVert \mathbf{\dot{u}}(\tau)	\rVert&\le V_{\mathrm{max}}, 0\le\tau\le \Delta t, \label{utv}
\end{align}
where $V_{\mathrm{max}}$ denotes the maximum speed of the UAV-PS. Note that (\ref{u0}) ensures that the UAV-PS starts and ends at the same point and that (\ref{utv}) imposes the practical maximal speed constraint. We discretize the duration $\Delta t$ into $N$ time slots indexed by $n$ with equal interval $\delta=\Delta t/N$. We assume that $\delta$ is sufficiently small such that the UAV-PS can be considered stationary in a time slot, even if it flies at the maximum speed $V_{\mathrm{max}}$. Therefore the trajectory constraints are rewritten as
\begin{align}
	\mathbf{u}[0]&=\mathbf{u}[N],\label{un}\\
	\lVert \mathbf{u}[n+1]-\mathbf{u}[n] \rVert ^2&\le (V_{\mathrm{max}}\delta)^2, n=0, 2, ..., N-1. \label{unv}
\end{align}

We assume the Doppler effect caused by the mobility of the UAV-PS can be well compensated at the receiver and perfect channel state information (CSI) is accessible at devices \cite{cao2020cooperative}. Further, due to the ground-to-air nature of UAV uplink communication, we consider the line-of-sight (LoS) channels for the device-UAV links. Thus the channel from the $m$-th device to the UAV-PS during the $n$-th time slot follows the free-space path loss model, expressed as
\begin{gather}
	h_m[n]=\sqrt{\varrho d_m^{-2}[n]}\vartheta_m[n], \label{channel}
\end{gather}
where $\varrho$ denotes the channel power gain at the reference distance $d_{\mathrm{ref}}=1$ m, $\vartheta_m[n]\triangleq e^{j\theta_m[n]}$ with $\theta_m[n]$ denoting the phase of $h_m[n]$,
and $d_m[n]=\sqrt{z^2+\|\mathbf{u}[n]-\mathbf{v}_m\|^2}$ denotes the distance between the device $m$ and the UAV-PS at the $n$-th time slot.

\subsection{Transmitting Signal Model}

As mentioned before, the UAV-PS only serves nearby devices. We now define ``nearby'' precisely. At the $n$-th time slot, given a preset coverage distance threshold $d_{\mathrm{thr}}$, we define set $\mathcal{M}[n]\triangleq \{m \in [M] \ |\  \|\mathbf{u}[n]-\mathbf{v}_m\|\le d_{\mathrm{thr}}\}$. The devices with indexes in $\mathcal{M}[n]$ are referred to as the nearby devices.
	
At the $n$-th time slot, only the nearby devices need to transmit signals. Assume that device $m$ is a nearby device. We now introduce how device $m$ generates its transmitting signal.
It first normalizes the local gradient $\mathbf{g}_m^{(t)}$ for efficient transmission.
Specifically, each entry of the normalized gradient $\tilde{\mathbf{g}}_m^{(t)}$ is computed by
\begin{gather}
	\tilde{g}_m^{(t)}(d)=\frac{g_m^{(t)}(d)-\bar{g}_m^{(t)}}{\sqrt{\upsilon_m^{(t)}}}, d\in [D],
	\label{normar}
\end{gather}
where $g_m^{(t)}(d)$ denotes the $d$-th entry of $\mathbf{g}_m^{(t)}$, $\bar{g}_m^{(t)}$ and $\upsilon_m^{(t)}$ are computed by $\bar{g}_m^{(t)}=\sum_{d=1}^D g_m^{(t)}(d)/D$ and $\upsilon_m^{(t)}=\sum_{d=1}^D(g_m^{(t)}(d)-\bar{g}_m^{(t)})^2/D$ respectively. Since the cost of transmitting the scalars $\{\bar{g}_m^{(t)}, \upsilon_m^{(t)}\}$ is negligible compared to that of transmitting the gradients, we assume that these scalars are sent to the UAV-PS losslessly following \cite{lin2021deploying} and \cite{yang2022federated}.


The normalized gradient $\tilde{\mathbf{g}}_m^{(t)}$ is then modulated as a complex vector $\mathbf{r}_{m}^{(t)}$ by
\begin{gather}
	\mathbf{r}_{m}^{(t)} \triangleq 
	\tilde{\mathbf{g}}_{m}^{(t)}\left (1:C \right ) + 
	j \tilde{\mathbf{g}}_{m}^{(t)}\left ((C+1) :2C\right ) \in \mathbb{C}^{C }, \label{pin}
\end{gather}
where $C=D/2$ and $\tilde{\mathbf{g}}_{m}^{(t)}\left (c_1:c_2 \right)$ denotes the sub-vector of $\tilde{\mathbf{g}}_{m}^{(t)}$ containing the entries with indexes from $c_1$ to $c_2$.\footnote{We assume an even $D$ for simplicity.} Then the transmitting signal is given by
\begin{gather}
	\mathbf{x}_m^{(t)}[n]\triangleq \beta_m[n]\mathbf{r}_{m}^{(t)}, \label{fa}
\end{gather}
where $\beta_m[n]\in\mathbb{C}$ is the transmitting coefficient satisfying the power constraints $|\beta_m[n]|^2 = P_0$.\footnote{We assume that all the devices transmit with full power, which is reasonable since the mobility of the UAV-PS allows the signals to be approximately aligned with desired coefficients (by designing an appropriate trajectory).}

\subsection{Hierarchical Aggregation}
\label{Hierarchical_Over-the-Air_Model_Aggregation}

In our hierarchical aggregation scheme for the UAV-assisted OA-FL, the local gradients are aggregated in two steps, called \emph{over-the-air partial aggregation} and \emph{global aggregation}. In this subsection, we separately introduce the two steps.

\subsubsection{Over-the-air partial aggregation}
In this step, the UAV-PS flies around to receive the partially aggregated local gradients for $N$ times. Specifically, the received signal at the $n$-th time slot is given by
\begin{align}
	\mathbf{y}^{(t)}[n]=\sum_{m\in \mathcal{M}[n]} h_m[n]\mathbf{x}_m^{(t)}[n]+\mathbf{n}^{(t)}[n], \label{yhx}
\end{align}
where $\mathbf{n}^{(t)}[n]\in \mathbb{C}^{C}$ is the additive white Gaussian noise (AWGN) with elements independently drawn from $\mathcal{CN}(0,\sigma^2)$.

For convenience, we introduce a binary indicator variable $\alpha_m[n] \triangleq \mathbbm{1}_{\|\mathbf{u}[n]-\mathbf{v}_m\|\le d_{\mathrm{thr}}}(\mathbf{u}[n])$, where the indicator function $\mathbbm{1}_{\|\mathbf{u}[n]-\mathbf{v}_m\|\le d_{\mathrm{thr}}}(\cdot)$ takes $1$ when $\mathbf{u}[n]$ satisfies $\|\mathbf{u}[n]-\mathbf{v}_m\|\le d_{\mathrm{thr}}$. Then the received signal can be rewritten as
\begin{align}
	\label{received_signal}
	\mathbf{y}^{(t)}[n]=\sum_{m=1}^{M} \alpha_m[n]h_m[n]\mathbf{x}_m^{(t)}[n]+\mathbf{n}^{(t)}[n].
\end{align}

\vspace{-0.1cm}
By substituting \eqref{channel} and \eqref{fa} into \eqref{received_signal}, we have
\begin{align}
	\mathbf{y}^{(t)}[n]=\sum_{m=1}^ {M} \frac{\sqrt{\varrho} \alpha_m[n] \beta_{m}[n] \vartheta_m[n] }{\sqrt{z^2+\|\mathbf{u}[n]-\mathbf{v}_m\|^2}} \mathbf{r}_{m}^{(t)}[n]	+ \mathbf{n}^{(t)}[n].
	\label{rhoppr}
\end{align}

Further, we set $\beta_m[n]=\overline{\vartheta_m[n]}\sqrt{P_0}$ to compensate for the phase shift of the channel. Thus the partially aggregated signal can finally be rewritten as
\vspace{-0.1cm}
\begin{gather}
	\mathbf{y}^{(t)}[n]=\sum_{m=1}^ {M} \frac{\sqrt{\varrho P_0} \alpha_m[n]} {\sqrt{z^2+\|\mathbf{u}[n]-\mathbf{v}_m\|^2}} \mathbf{r}_{m}^{(t)}[n]	+ \mathbf{n}^{(t)}[n]. \label{yfinal}
\end{gather}

\subsubsection{Global aggregation}
After the UAV-PS completes one round, it aggregates $\{\mathbf{y}^{(t)}[n]\}_{n=1}^N$ using coefficients $\{\zeta[n]\in\mathbb{R}\}_{n=1}^N$ to obtain globally aggregated signal $\mathbf{a}^{(t)}$, i.e.,
\begin{gather}
	\mathbf{a}^{(t)}=\sum_{n=1}^ {N}\zeta[n]\mathbf{y}^{(t)}[n]. \label{ay}
\end{gather}

\vspace{-0.1cm}
Then $\mathbf{\hat{g}}^{(t)}$ is computed by
\vspace{-0.08cm}
\begin{gather}
	\mathbf{\hat{g}}^{(t)} = \left[ \Re\{\mathbf{a}^{(t)}\}^\top , \, \Im\{\mathbf{a}^{(t)}\}^\top \right]^\top + \bar{g}^{(t)}\mathbf{1}_{D  }, \label{chai}
\end{gather}
where $ \bar{g}^{(t)}\triangleq\sum_{m=1}^M b_m\bar{g}_m^{(t)}$.

\section{Performance Analysis}

For the convenience of convergence analysis, we first introduce two standard assumptions \cite{friedlander2012hybrid, liu2021reconfigurable}:
\begin{assumption}
	\label{assum1}
	The global loss function $F$ is strongly convex with positive parameter $\mu$: $F(\mathbf{w}) 
	\geq F(\mathbf{w^\prime}) + (\mathbf{w} - \mathbf{w^\prime})^{\top} \nabla F(\mathbf{w^\prime}) + \frac{\mu}{2} \| \mathbf{w} - \mathbf{w^\prime}\|^2, \forall \mathbf{w}, \mathbf{w^\prime} \in \mathbb{R}^D.$
\end{assumption}
\begin{assumption}
	\label{assum2}
	The 
	gradient $\nabla F(\cdot)$ is uniformly Lipschitz continuous with parameter $\omega$, i.e., $\| \nabla F(\mathbf{w}) - \nabla F(\mathbf{w^\prime})\| \leq \omega \|\mathbf{w} - \mathbf{w^\prime} \|, \forall \mathbf{w}, \mathbf{w^\prime} \in \mathbb{R}^D$.
\end{assumption}

	Following \cite{zhong2021over}, we further assume the following gradient correlation model:
	\begin{assumption}
		\label{assum4}
		(Gradient correlation model) Define $\tilde{\mathbf{G}}^{(t)} \triangleq [\tilde{\mathbf{g}}_{1}^{(t)}, \cdots, \tilde{\mathbf{g}}_{M} ^{(t)}] \in \mathbb{R}^{D \times M}$. Let $\mathbf{z}_{d}^{(t)\top} \in \mathbb{R}^{M}$ denote the $d$-th row of $\tilde{\mathbf{G}}^{(t)}$, $\forall d \in [D]$. We assume $\{\mathbf{z}_{d}^{(t)}\}_{d=1}^D$ to be independent and identically distributed (i.i.d.) with $\mathbb{E}[\mathbf{z}_{d}^{(t)}] = \mathbf{0}$ and $\mathbb{E}[\mathbf{z}_{d}^{(t)}(\mathbf{z}_{d}^{(t)})^\top] = \boldsymbol{\rho}^{(t)}$, $\forall d \in [D]$.
	\end{assumption}
	
	An upper bound of the convergence performance is given in the following theorem.
	\vspace{-0.15cm}
	\begin{theorem}
		\label{Theorem:Convergence}
		Set the learning rate $\eta=1/\omega$. Under Assumptions \ref{assum1} and \ref{assum2}, we have:
		\vspace{-0.15cm}
		\begin{gather}
			\label{aaa}
			\mathbb{E}\!\left[\!\mathcal{F}\!\left(\!\mathbf{w}^{(T+1)}\!\right)\!- \!\mathcal{F}\left(\mathbf{w}^*\right)\right]
			\leq \sum_{t=0}^{T}\left(1-\frac{\mu}{\omega}\right)^{T-t}\mathbb{E}[\|\mathbf{e}^{(t)}\|^2]  \notag\\
			+\left(\mathcal{F}\left(\mathbf{w}^{(0)}\right) - \mathcal{F}\left(\mathbf{w}^*\right)\right) \left(1-\frac{\mu}{\omega}\right)^{T+1}, 
		\end{gather}
		where $\mathbf{e}^{(t)}\triangleq \mathbf{g}^{(t)} - \hat{\mathbf{g}}^{(t)}$ and the expectations are taken with respect to the gradients and the AWGN.
	\end{theorem}
	\begin{proof}
		Take expectation of both sides of the recursion in [\citenum{friedlander2012hybrid}, Lemma~2.1] and iterate it for $t+1$ times.
	\end{proof}
	
	Clearly, by Theorem \ref{Theorem:Convergence}, to improve the learning performance, we only have to minimize the mean squared error (MSE) $\mathbb{E}[\|\mathbf{e}^{(t)}\|^2]$, which can be written as
	\begin{align}
		&\mathbb{E}[\|\mathbf{e}^{(t)}\|^2]
		\notag\\
		=&\mathbb{E}\left[\left\|\mathbf{g}^{(t)} -  \hat{\mathbf{g}}^{(t)}\right\|^2\right]
		\notag \\
		\overset{\text{(a)}}{=}& \mathbb{E}\left[\left\|\sum_{m=1}^{M}b_m g_m^{(t)} -  \sum_{m=1}^{M}b_m\bar{g}^{(t)}\mathbf{1}_{D} - \begin{bmatrix}
			\Re\{\mathbf{a}^{(t)}\}\\
			\Im\{\mathbf{a}^{(t)}\}
		\end{bmatrix} \right\|^2\right]
		\notag \\
		\overset{\text{(b)}}{=}&\mathbb{E}\left[\left\|\sum_{m=1}^{M}\left( b_m\sqrt{\upsilon^{(t)}_m} - \sum_{n=1}^{N}	\frac{\zeta[n]\sqrt{\varrho P_0} \alpha_m[n]} {\sqrt{z^2+\|\mathbf{u}[n]-\mathbf{v}_m\|^2}}  \right) \tilde{\mathbf{g}}_m^{(t)} 
		\right.\right. \notag \\
		&\left.\left. - \sum_{n=1}^{N}\zeta[n]\mathbf{n}_{r}^{(t)}[n] \right\|^2\right],
		\label{Ee}
	\end{align}
	where $\mathbf{n}_{r}^{(t)}[n] \triangleq [\Re\{\mathbf{n}^{(t)}[n]\}^\top, \Im\{\mathbf{n}^{(t)}[n]\}^\top]^\top$, step ($a$) follows from \eqref{idealtri} and \eqref{chai}, and step ($b$) follows from \eqref{normar}, \eqref{pin}, \eqref{yfinal} and \eqref{ay}.
	
	To proceed, define $\boldsymbol{\zeta}\triangleq\left[\zeta[1], ..., \zeta[N]\right]^\top$, $\boldsymbol{\upsilon}^{(t)}\triangleq\Big[b_{1}\sqrt{\upsilon^{(t)}_{1}} , ..., b_{M}\sqrt{\upsilon^{(t)}_{M}}\Big]^\top$, and $\mathbf{K}\in \mathbb{R}^{M\times N}$ with $K_{m,n}=\dfrac{\alpha_m[n]\sqrt{\varrho P_0}}{\sqrt{z^2+\|\mathbf{u}[n]-\mathbf{v}_{m}\|^2}}$. Recall $\tilde{\mathbf{G}}^{(t)} = [\tilde{\mathbf{g}}_{1}^{(t)}, \cdots, \tilde{\mathbf{g}}_{M} ^{(t)}]$. Then \eqref{Ee} can be rewritten as
	\begin{align}
		&\mathbb{E}[\|\mathbf{e}^{(t)}\|^2]
		\notag\\
		=&\mathbb{E}\left[\left\|\tilde{\mathbf{G}}^{(t)} \left(\boldsymbol{\upsilon}^{(t)}- \mathbf{K} \boldsymbol{\zeta}\right) -  \sum_{n=1}^{N}\zeta[n]\mathbf{n}_{r}^{(t)}[n]\right\|^2\right]
		\notag \\
		\overset{\text{(a)}}{=}& D\left(\boldsymbol{\upsilon}^{(t)}- \mathbf{K} \boldsymbol{\zeta}\right)^\top \boldsymbol{\rho}^{(t)} \left(\boldsymbol{\upsilon}^{(t)}- \mathbf{K} \boldsymbol{\zeta}\right) +\frac{D\sigma^2}{2}\boldsymbol{\zeta}^\top\boldsymbol{\zeta},
		\label{ana2}
	\end{align}
	where step ($a$) follows from Assumption \ref{assum4} and the fact that the AWGN is independent of the gradients.

	\section{System optimization}
	
	\subsection{Problem Formulation}
	Define $\mathcal{U}=\{\mathbf{u}[n]\}_{n=1}^N$. The MSE minimization problem can be formulated as
	\begin{subequations} \label{P0}
		\begin{align}
			\min_{\boldsymbol{\zeta}, \mathcal{U}}  \quad & 
			\frac{\sigma^2}{2}\boldsymbol{\zeta}^\top\boldsymbol{\zeta}
			+\left(\boldsymbol{\upsilon}^{(t)}- \mathbf{K} \boldsymbol{\zeta}\right)^\top \boldsymbol{\rho}^{(t)} \left(\boldsymbol{\upsilon}^{(t)}- \mathbf{K} \boldsymbol{\zeta}\right) \\
			\operatorname{ s.t. } \quad
			&\eqref{un}, \ \eqref{unv}. 
		\end{align}
	\end{subequations}
	
	To properly solve problem \eqref{P0}, we first introduce $\mathbf{K}$ as an optimization variable, and then slack the consequent equality constraints. By introducing $\mathbf{K}$, problem \eqref{P0} can be reformulated as
	\begin{subequations} \label{P01}
		\begin{align}
			\min_{\boldsymbol{\zeta}, \mathbf{K}, \mathcal{U}}  \quad & 
			\frac{\sigma^2}{2}\boldsymbol{\zeta}^\top\boldsymbol{\zeta}
			+\left(\boldsymbol{\upsilon}^{(t)}- \mathbf{K} \boldsymbol{\zeta}\right)^\top \boldsymbol{\rho}^{(t)} \left(\boldsymbol{\upsilon}^{(t)}- \mathbf{K} \boldsymbol{\zeta}\right) 
			\label{p00} \\
			\operatorname{ s.t. } \quad
			&K_{m,n}\!=\!\dfrac{\alpha_m[n]\sqrt{\varrho P_0}}{\sqrt{z^2+\|\mathbf{u}[n]-\mathbf{v}_{m}\|^2}}, \forall m \in [M], n\in [N],
			\label{constraint_01}\\
			&\eqref{un}, \ \eqref{unv}.
		\end{align}
	\end{subequations}
	Note that the optimal $\boldsymbol{\zeta}^*$ and $\mathcal{U}^*$ for \eqref{P01} are also optimal for \eqref{P0}. We then relax constraints \eqref{constraint_01} by slacking $\alpha_m[n]$. Recall $\alpha_m[n]=\mathbbm{1}_{\|\mathbf{u}[n]-\mathbf{v}_m\|\le d_{\mathrm{thr}}}(\mathbf{u}[n])$, which is a binary variable. Thus a natural way to slack $\alpha_m[n]$ is to replace the constraint $\alpha_m[n]=\mathbbm{1}_{\|\mathbf{u}[n]-\mathbf{v}_m\|\le d_{\mathrm{thr}}}(\mathbf{u}[n])$ by
	\begin{gather}
		0\le\alpha_m[n]\le\frac{d_{\mathrm{thr}}^2}{d_{\mathrm{thr}}^2+\|\mathbf{u}[n]-\mathbf{v}_m\|^2}.
	\end{gather}
	Constraints \eqref{constraint_01} thus can be relaxed as
	\begin{gather}\label{K_slacking}
		0\!\le\! K_{m,n}\!\le\!\frac{d_{\mathrm{thr}}^2\sqrt{\varrho P_0}} {(d_{\mathrm{thr}}^2+\|\mathbf{u}[n]-\mathbf{v}_m\|^2)(z^2+\|\mathbf{u}[n]-\mathbf{v}_{m}\|^2)^{\tfrac{1}{2}}},
	\end{gather}
	resulting in the following optimization problem:
	\begin{subequations} \label{P1}
		\begin{align}
			\min_{\boldsymbol{\zeta}, \mathbf{K}, \mathcal{U}} \quad & \!\!\!
			\frac{\sigma^2}{2}\boldsymbol{\zeta}^\top\boldsymbol{\zeta}
			+\left(\boldsymbol{\upsilon}^{(t)}- \mathbf{K} \boldsymbol{\zeta}\right)^\top\!\!\! \boldsymbol{\rho}^{(t)}\!\! \left(\boldsymbol{\upsilon}^{(t)}- \mathbf{K} \boldsymbol{\zeta}\right)
			\label{std0}\\
			\operatorname{ s.t. } \quad
			&\eqref{un},\  \eqref{unv},\ \eqref{K_slacking}. 
		\end{align}
	\end{subequations}

	The objective of problem \eqref{P1} is a multi-convex function of $\boldsymbol{\zeta}$ and $\mathbf{K}$, which inspires us to apply the alternating optimization (AO) technique.


	
	\subsection{Optimization of $\{\mathbf{K}, \mathcal{U}\}$ with a fixed $\boldsymbol{\zeta}$}
	
	In this subsection, we optimize $\mathbf{K}$ and $\mathcal{U}$ with a fixed $\boldsymbol{\zeta}=\tilde{\boldsymbol{\zeta}}$. The corresponding optimization problem is
	\begin{subequations} \label{P21}
		\begin{align}
			\min_{\mathbf{K}, \mathcal{U}}  \quad &  \tilde{\boldsymbol{\zeta}}^\top\mathbf{K}^\top \boldsymbol{\rho}^{(t)} \mathbf{K}\tilde{\boldsymbol{\zeta}} - 2 \boldsymbol{\upsilon}^{(t)\top} \boldsymbol{\rho}^{(t)} \mathbf{K}\tilde{\boldsymbol{\zeta}} \\
			\operatorname{ s.t. } \quad
			&\eqref{un}, \ \eqref{unv}, \ \eqref{K_slacking}. 
		\end{align}
	\end{subequations}
	Note that constraint \eqref{K_slacking} is non-convex for $\mathbf{u}[n]$. However, the right side of \eqref{K_slacking} is convex with respect to the whole expression $\|\mathbf{u}[n]-\mathbf{v}_{m}\|^2$. Thus we can obtain a lower bound of the right side of \eqref{K_slacking} by expanding it as its first-order Taylor polynomial with respect to $\|\mathbf{u}[n]-\mathbf{v}_{m}\|^2$ \cite{boyd2004convex} based on the idea of successive convex approximation (SCA), i.e.,
	\begin{align}
		\label{ss}
		&\frac{d_{\mathrm{thr}}^2\sqrt{\varrho P_0}} {(d_{\mathrm{thr}}^2+\|\mathbf{u}[n]-\mathbf{v}_m\|^2)(z^2+\|\mathbf{u}[n]-\mathbf{v}_{m}\|^2)^{\tfrac{1}{2}}}
		\ge 
		\notag \\
		&\ \ \ \ \  \Psi^{(t)} + \Psi'^{(t)}
		\left(\|\mathbf{u}[n]-\mathbf{v}_m\|^2-\|\tilde{\mathbf{u}}[n]-\mathbf{v}_m\|^2\right),
	\end{align}	
	where
	\begin{align*}
		\Psi^{(t)}=\frac{d_{\mathrm{thr}}^2 \sqrt{\varrho P_0}} {(d_{\mathrm{thr}}^2+\|\tilde{\mathbf{u}}[n]-\mathbf{v}_m\|^2)(z^2+\|\tilde{\mathbf{u}}[n]-\mathbf{v}_{m}\|^2)^{\tfrac{1}{2}}}, \\
		\Psi'^{(t)}=\frac{-d_{\mathrm{thr}}^2 \left(1+\frac{1}{2}d_{\mathrm{thr}}^2 + \frac{3}{2} \|\tilde{\mathbf{u}}[n]-\mathbf{v}_{m}\|\right) \sqrt{\varrho P_0} } {(d_{\mathrm{thr}}^2+\|\tilde{\mathbf{u}}[n]-\mathbf{v}_m\|^2)^2(z^2+\|\tilde{\mathbf{u}}[n]-\mathbf{v}_{m}\|^2)^{\tfrac{3}{2}}},
	\end{align*}
	and $\|\tilde{\mathbf{u}}[n]-\mathbf{v}_{m}\|^2$ denotes the expand point (i.e., expanding at point $\|\tilde{\mathbf{u}}[n]-\mathbf{v}_{m}\|^2$). The equality in \eqref{ss} holds when $\mathbf{u}[n] = \tilde{\mathbf{u}}[n]$. Then \eqref{K_slacking} can be transformed into
	\begin{gather}
		0\le K_{m,n}\le\Psi^{(t)} + \Psi'^{(t)}\!	\left(\|\mathbf{u}[n]-\mathbf{v}_m\|^2-\|\tilde{\mathbf{u}}[n]-\mathbf{v}_m\|^2\right).
		\label{kfinal}
	\end{gather}
	Since $\|\mathbf{u}[n]-\mathbf{v}_{m}\|^2$ is a convex function of $\mathbf{u}[n]$ and the right side is a linear function of $\|\mathbf{u}[n]-\mathbf{v}_{m}\|^2$, constraint \eqref{kfinal} is convex. Finally, problem \eqref{P21} can be reformulated into
	\begin{subequations} \label{P2}
		\begin{align}
			\min_{\mathbf{K}, \mathcal{U}}  \quad &  \tilde{\boldsymbol{\zeta}}^\top\mathbf{K}^\top \boldsymbol{\rho}^{(t)} \mathbf{K}\tilde{\boldsymbol{\zeta}} - 2 \boldsymbol{\upsilon}^{(t)\top} \boldsymbol{\rho}^{(t)} \mathbf{K}\tilde{\boldsymbol{\zeta}} \\
			\operatorname{ s.t. } \quad
			&\eqref{un}, \ \eqref{unv}, \ \eqref{kfinal},
		\end{align}
	\end{subequations}
	which is convex and can be efficiently solved by standard convex optimization solvers such as CVX.
	
	\subsection{Optimization of $\boldsymbol{\zeta}$ with fixed $\mathbf{K}$ and $\mathcal{U}$}
	\label{optizeta}
	
	Since \eqref{std0} is convex with respect to $\boldsymbol{\zeta}$ with fixed $\mathbf{K} = \tilde{\mathbf{K}}$ and $\mathcal{U} = \tilde{\mathcal{U}}$, by letting $\dfrac{\partial\mathbb{E}[\|\mathbf{e}^{(t)}\|^2]}{\partial\boldsymbol{\zeta}}=0$, the optimal $\boldsymbol{\zeta}^*$ is given by
	\begin{gather}
		\boldsymbol{\zeta}^*= \left(\frac{\sigma^2}{2}\mathbf{I}_{N\times N} + \tilde{\mathbf{K}}^\top  \boldsymbol{\rho}^{(t)}  \tilde{\mathbf{K}} \right)^{-1}  \tilde{\mathbf{K}}^\top  \boldsymbol{\rho}^{(t)} \boldsymbol{\upsilon}^{(t)}. \label{zetaopt}
	\end{gather}
	
	\subsection{Overall Algorithm}
	By alternatively optimizing $\{\mathbf{K}, \mathcal{U}\}$ and $\boldsymbol{\zeta}$, problem \eqref{P1} can be solved, and thus problem \eqref{P0}. We summarize the proposed algorithm in Algorithm \ref{algorithm}. This algorithm converges since the objective \eqref{std0} is monotonically non-decreasing in the iterative process.
	\begin{algorithm} [t!]
		\caption{Algorithm for Solving Problem \eqref{P0}}\label{algorithm}
		\begin{algorithmic}[1]
			\State Input: $\boldsymbol{\rho}^{(t)}$ and $\boldsymbol{\upsilon}^{(t)}$.
			\State Initialize $\{\mathbf{u}^0[n]\}_{n=1}^N$ and $\boldsymbol{\zeta}^0$ to feasible values, and set $i=0$.
			\Repeat
			\State Set $\tilde{\boldsymbol{\zeta}} = \boldsymbol{\zeta}^i$ and $\tilde{\mathbf{u}}[n] = \mathbf{u}^i[n]$, $\forall n \in [N]$, then solve problem \eqref{P2} to obtain $\{\mathbf{u}^{i+1}[n]\}_{n=1}^N$ and $\mathbf{K}^{i+1}$.
			\State Set $\tilde{\mathbf{K}} = \mathbf{K}^{i+1}$, then $\boldsymbol{\zeta}^{i+1}= (\frac{\sigma^2}{2}\mathbf{I}_{N\times N} + \tilde{\mathbf{K}}^\top  \boldsymbol{\rho}^{(t)}  \tilde{\mathbf{K}} )^{-1}  \tilde{\mathbf{K}}^\top  \boldsymbol{\rho}^{(t)} \boldsymbol{\upsilon}^{(t)}$.
			\State Update $i \leftarrow i+1 $.
			\Until The fractional decrease of the objective \eqref{std0} is below a threshold $\epsilon > 0$.
			\State Set $\boldsymbol{\zeta}^* = \boldsymbol{\zeta}^{i}$ and $\mathbf{u}^*[n]=\mathbf{u}^{i}[n]$, $\forall n \in [N]$.
		\end{algorithmic}
	\end{algorithm}

%

\section{Simulation Results}
In this section, we validate the effectiveness of the proposed UAV-PS assisted OA-FL scheme. We consider an FL system with $M=20$ devices on the ground in a large service area of $2\times2 $ km$^2$. The devices are randomly distributed in groups in order to simulate the users in clusters as in the village. The UAV-PS flies in an assumed fixed height $z=50$ m above the area and has a fixed starting and ending horizontal location $\mathbf{u}[0]=\mathbf{u}[N]=[885,-10]^\top$. The channel power gain, the noise power, and the maximum transmitting power are set as $\varrho=-60$ dB, $\sigma^2=-90$ dBm, and $P_0=0.32$ W. The maximum speed and the flying time interval of the UAV-PS are assumed as $V_{\mathrm{max}}=50$ m/s and $\delta=1$ s. The coverage distance threshold is set as $d_{\mathrm{thr}}=158$ m and after optimization the reconstruction to the binary $\alpha_m[n]$ follows the rounded principle based on the optimized results. And we set the optimization threshold $\epsilon=10^{-4}$.

The learning task is conducted over the MNIST dataset 
with $Q=60000$ training samples. We train a convolutional neural network (CNN) with two $5\times5$ convolution layers in which the first has $16$ channels, the second has $32$ channels, and both have a $2\times2$ max pooling, and with a fully connected layer of $50$ neurons and ReLu activation, and with a softmax output layer (total parameters $D=39408$). The loss function is the cross-entropy loss. The learning rate is set as $\eta=0.05$ with momentum $=0.5$ and the local updates include $5$ mini-batches of stochastic gradient descent (SGD). We conduct the experiments over $10$ Monte Carlo trials and $20$ Monte Carlo trials for i.i.d. datasets and non-i.i.d. datasets respectively. For the i.i.d. data condition, data samples are assigned evenly to all devices; for the non-i.i.d. data condition, each device randomly selects $5$ classes with $Q/5M$ samples for each selected class. The correlation matrices of the gradients are approximated by $\boldsymbol{\rho}^{(t)} = \frac{1}{D} \sum_{d=1}^D \mathbf{z}_{d}^{(t)} \mathbf{z}_{d}^{(t)}{}^\top$.

We first present the PS location/UAV-PS trajectory optimization results in the following simulation conditions: (1) Static PS without UAV assisted, which is located on the coordinate barycenter of all devices; (2) UAV-PS without optimization in a circular trajectory with the flying period $\Delta t=120$ s; (3) UAV-PS with optimization by Algorithm \ref{algorithm} and the flying period $\Delta t=120$ s; and (4) UAV-PS with optimization by Algorithm \ref{algorithm} and the flying period $\Delta t=80$ s. All the above simulations are on the i.i.d. datasets.
\begin{figure} 
	[!ht]
	\centering
	\vspace{-1.0em}
	\includegraphics[scale=0.29]{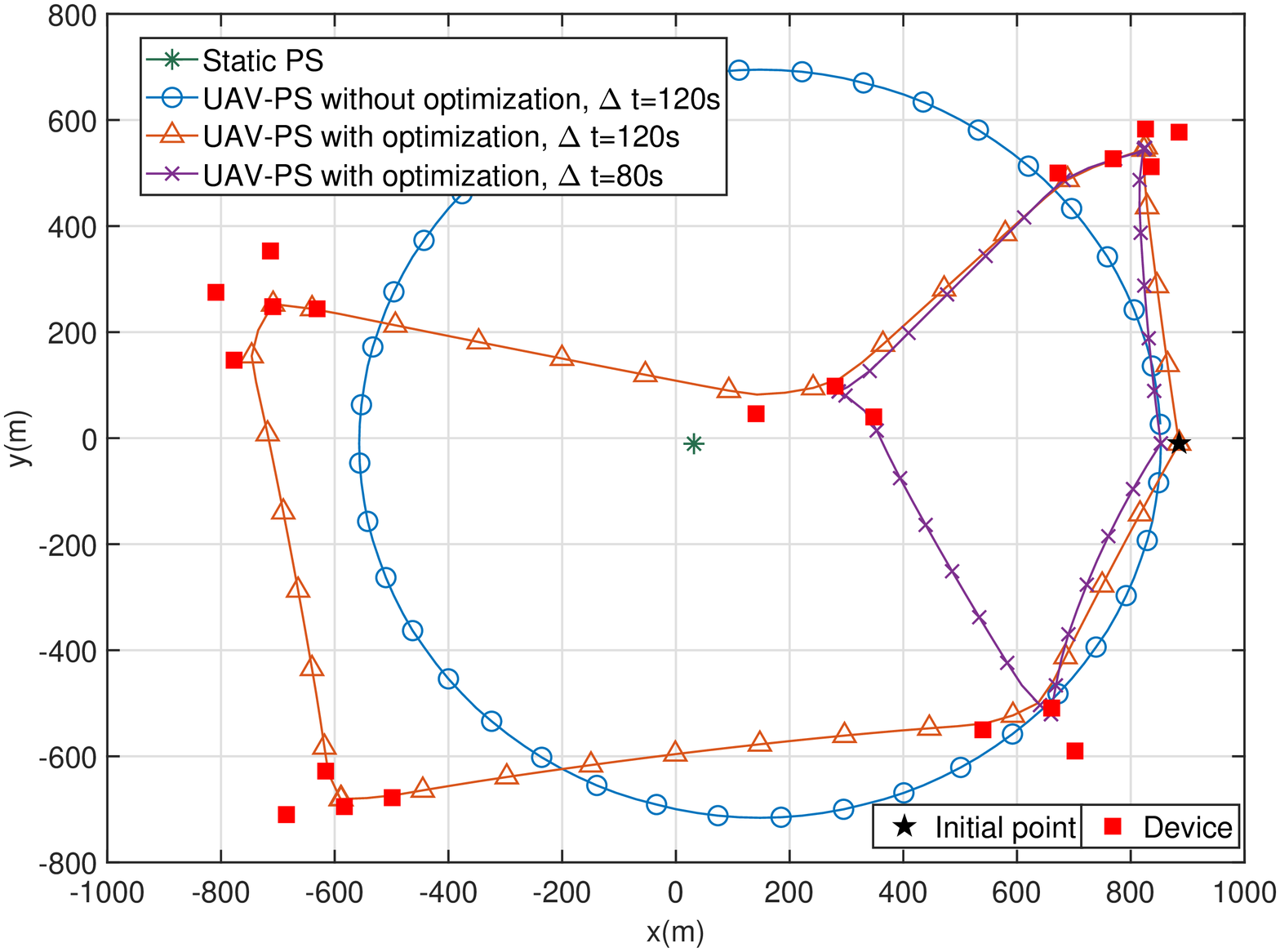}
	\caption{PS location/UAV-PS trajectory with $20$ devices.}
	\label{traj}
	\vspace{-0.2em}
\end{figure}

In Fig. \ref{traj}, by appropriate trajectory optimization, the UAV-PS with sufficient time, like $\Delta t=120$ s, has no need to fly right above each device but can serve for all on account of the AirComp technique to aggregate in clusters, which definitely verifies the effectiveness of the UAV assisted hierarchical over-the-air aggregation. However, when the flying time is limited, for instance, the UAV-PS with $\Delta t=80$ s, it fails to fly across the whole service area and thus only be optimized to fly above nearby clusters. 

\begin{figure} 
	[!ht]
	\centering
	\vspace{-1.0em}
	\includegraphics[scale=0.317]{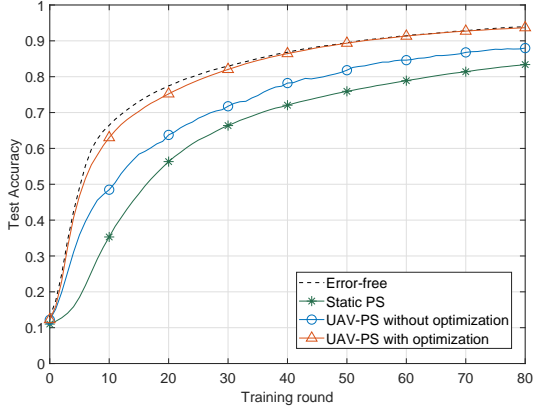}
	\caption{OA-FL test accuracy versus training rounds, i.i.d. data.}
	\label{accuiid}
	\vspace{-0.2em}
\end{figure}

\begin{figure} 
	[!ht]
	\centering
	\vspace{-1.0em}
	\includegraphics[scale=0.342]{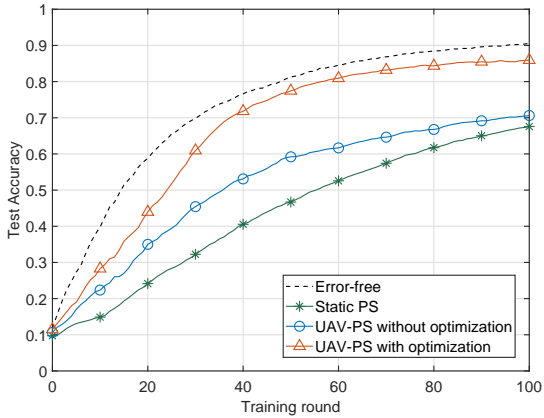}
	\caption{OA-FL test accuracy versus training rounds, non-i.i.d. data.}
	\label{accunon}
	\vspace{-0.2em}
\end{figure}

Figs. \ref{accuiid} and \ref{accunon} show the test accuracy of the UAV-assisted OA-FL system with i.i.d. data and non-i.i.d. data respectively in the following schemes: (1) Error-free bound with the PS aggregating free of error; (2) Static PS without UAV assisted, as the green location in Fig. \ref{traj}, whose limitation of coverage area is waived for serving all devices, and which conducts the aggregation design as \cite{zhong2021over} to relieve the straggler issue; (3) UAV-PS in a circular trajectory, as the blue line in Fig. \ref{traj}, $\Delta t=120$ s, which conducts the aggregation design optimization following \ref{optizeta}, and the empirical optimization for the radius and the center of the circle; and (4) UAV-PS with joint optimization for the trajectory and aggregation design by Algorithm \ref{algorithm}, as the orange line in Fig. \ref{traj},  $\Delta t=120$ s.

For the static PS though located on the barycenter, large communication distances for remote devices as the stragglers severely limit the quality of gradient aggregation, thus leading to bad learning results. The UAV-PS with trajectory optimization by Algorithm \ref{algorithm} is closer to the error-free baseline when compared to the one with a circular trajectory. This indicates that the latter, without the trajectory optimization for the UAV-PS, still constrains the training effect. 
And in Fig. \ref{accuiid} and \ref{accunon}, our scheme performs far better than any others on both i.i.d. and non-i.i.d. datasets, which demonstrates the effectiveness of the scheme. Further, we discover that on non-i.i.d. datasets, the learning accuracy is of higher dependence on communication and aggregation precision. Therefore our scheme is proven to have great potential on complex datasets and in complex learning scenarios.

%

\section{Conclusion}
In this paper, we proposed a UAV-PS assisted hierarchical aggregation for over-the-air FL in large service area scenarios and analyzed the convergence performance of the proposed system with gradient correlation in consideration. We formulated an MSE optimization problem to jointly optimize the UAV trajectory and the aggregation coefficients. An algorithm based on AO and SCA was proposed to solve the problem. Numerical results demonstrated the effectiveness of our scheme.


\bibliographystyle{IEEEtran}
\bibliography{IEEEabrv,mybib}

\end{document}